\documentclass[copyright,creativecommons]{eptcs}
\providecommand{\event}{HOR 2010}

\usepackage[latin1]{inputenc}
\usepackage{alltt}
\usepackage{xspace}
\usepackage{amsthm}
\usepackage{stmaryrd}
\usepackage{bussproofs}

\usepackage{tikz}
\usetikzlibrary{positioning}

\usepackage{float}
\floatstyle{boxed} 
\restylefloat{figure}

\newtheorem{theorem}{Theorem}
\newtheorem{example}{Example}
\newtheorem{lemma}{Lemma}

\newcommand{\PPC}{\emph{PPC}\xspace}
\newcommand{\EPPC}{\emph{$PPC_{EM}$}\xspace}
\newcommand{\DSPPC}{\emph{$PPC_\bullet$}\xspace}
\newcommand{\fail}{\ensuremath{\bot}\xspace}
\newcommand{\wait}{\texttt{wait}\xspace}
\newcommand{\case}[3][\theta]{[#1]#2\shortrightarrow#3}
\newcommand{\cmatch}[3][\theta]{\{\!\!\!\{#2\,/_{\!#1}\ #3\}\!\!\!\}}
\newcommand{\match}[3][\theta]{\{#2\,/_{\!#1}\ #3\}}
\newcommand{\re}{\longrightarrow}
\newcommand{\rre}{\longleftarrow}
\newcommand{\reppc}{\re_{PPC}}
\newcommand{\reeppc}{\re_{EM}}
\newcommand{\reb}{\re_{\bullet}}
\newcommand{\reB}{\re_{B}}
\newcommand{\rem}{\re_{m}}
\newcommand{\rer}{\re_{r}}
\newcommand{\rep}{\re_{p}}
\newcommand{\Pre}{\Longrightarrow}
\newcommand{\rPre}{\Longleftarrow}
\newcommand{\matching}[3][\theta]{\left\langle#1|#2|#3\right\rangle}
\newcommand{\sem}[3][\theta]{\left\llbracket#1|#2|#3\right\rrbracket}

\newcommand{\set}[1]{\{#1\}}
\newcommand{\fv}[1]{f\!v(#1)}
\newcommand{\fm}[1]{f\!m(#1)}
\newcommand{\fn}[1]{f\!n(#1)}
\newcommand{\trad}[1]{\llbracket#1\rrbracket}

\newcommand{\pure}{\!\!\downarrow\,}
\newcommand{\size}[1]{\ensuremath{\mathcal{S}(#1)}\xspace}
\newcommand{\nestor}{\ensuremath{\prec_\mathcal{N}}\xspace}
\newcommand{\edward}{\ensuremath{\prec_\mathcal{S}}\xspace}

\title{On the Implementation of Dynamic Patterns} \author{Thibaut Balabonski
  \institute{Laboratoire PPS, CNRS and Université Paris Diderot}
  \email{thibaut.balabonski@pps.jussieu.fr} }

\def\titlerunning{On the Implementation of Dynamic Patterns}
\def\authorrunning{T. Balabonski}

\begin{document}
\maketitle

\begin{abstract}
  The evaluation mechanism of pattern matching with dynamic patterns
  is modelled in the {\em Pure Pattern Calculus} by one single meta-rule. This
  contribution presents a refinement which narrows the gap between the
  abstract calculus and its implementation. A calculus is designed to
  allow reasoning on matching algorithms. The new calculus is proved to be
  confluent, and to simulate the original {\em Pure Pattern Calculus}. A family
  of new, matching-driven, reduction strategies is proposed.
\end{abstract}

\section*{Introduction: Dynamic Patterns}

{\bf Pattern matching} is a basic mechanism used to deal with algebraic data
structures in functional programming languages. It allows to define a function
by reasoning on the shape of the arguments.  For instance, define a
binary tree to be either a single data or a node with two subtrees (code on the left,
in ML-like syntax). Then a function on binary trees may be defined by reasoning
on the shapes generated by these two possibilities (code on the right).
\medskip\\
{\small\begin{minipage}{0.49\linewidth}
\begin{alltt}
  type 'a tree =
          | Data 'a
          | Node of 'a tree * 'a tree
\end{alltt}
\end{minipage}
\begin{minipage}{0.49\linewidth}
\begin{alltt}
let f t = match t with
        | Data d           ->  <code1>
        | Node (Data d) r  ->  <code2>
        | Node l r         ->  <code3>
\end{alltt}
\end{minipage}}
\medskip\\
An argument given to the function {\tt f} is first compared to (or {\bf matched}
against) the shape {\tt Data} {\tt d} (called a {\bf pattern}). In case of
success, the occurrences of {\tt d} in {\tt <code1>} are replaced by the
corresponding part of
the argument, and {\tt <code1>} is executed. In case of failure of this first
matching (the argument is not a data) the argument is matched against the second
pattern, and so on until a matching succeeds or there is no pattern left.

One limit of this approach is that patterns are fixed expressions
mentioning explicitly the constructors to which they can apply, which restricts
polymorphism and reusability of the code.  This can be improved by allowing
patterns to be parametrised: one single function can be specialised in various
ways by instantiating the parameters of its patterns by different constructors
or even by functions building patterns. For instance in the following
code, the function {\tt f} would take an additional parameter {\tt\bf p} which
would then be used to define the first two patterns. In this case, instantiating
{\tt\bf p} with the constructor {\tt Data} would yield the same function as
before, but any other function building a pattern can be used for {\tt\bf p}!

\begin{center}
\begin{minipage}{0.5\linewidth}
\begin{alltt}
let f {\bf p} t = match t with
          | {\bf p} d            ->  <code1>
          | Node ({\bf p} d) r   ->  <code2>
          | Node l r     {\bf }  ->  <code3>
\end{alltt}
\end{minipage}
\end{center}

However, introducing parameters and functions inside patterns deeply modifies
their nature: they become dynamic objects that have to be evaluated. This
disrupts the matching algorithms and introduces new evaluation behaviours. This
paper intends to give tools to study these extended evaluation possibilities.

The {\em Pure Pattern Calculus} (\PPC) of
B. Jay and D. Kesner~\cite{PPC,PatternCalculus} models the behaviour of dynamic
patterns by using a
meta-level notion of pattern matching. The present contribution analyses the
content of the meta pattern matching of \PPC (reviewed in Section~\ref{PPC}),
and proposes an explicit pattern matching calculus (Section~\ref{EPPC}) which is
confluent, which simulates \PPC, and which allows the description of new
reduction strategies (Section~\ref{Applications}). An extension of the explicit
calculus is then discussed (Section~\ref{EPMPS}) before a conclusion is drawn.

\section{The Pure Pattern Calculus}
\label{PPC}

This section only reviews some key aspects of \PPC. Please refer to~\cite{PPC}
for a complete story with more examples.  The syntax of \PPC is close to
the one of $\lambda$-calculus. The main difference is the replacement of the
abstraction over a variable $\lambda x.b$ by an abstraction over a pattern (with
a list of matching variables) written $\case{p}{b}$. There is also a new
distinction between {\bf variable} occurrences $x$ and {\bf matchable}
occurrences $\hat{x}$ of a name $x$. Variable occurrences are usual variables
which may be substituted while matchable occurrences are immutable and used as
matching variables or constructors.

\begin{center}
\begin{tabular}{r@{$\quad::=\quad$}l@{\hspace{1.6cm}}l}
$t$ & $x\ |\ \hat{x}\ |\ tt\ |\ \case{t}{t}$ & \PPC Terms
\end{tabular}
\end{center}

\noindent
where $\theta$ is a list of names. Letter $a$ (resp. $b$, $p$) is used to
indicate a term in position of {\bf a}rgument (resp. function {\bf b}ody,
{\bf p}attern).

As pictured below, in the abstraction $\case{p}{b}$ the list of names $\theta$
binds matchable occurrences in the pattern $p$ and variable occurrences in the
body $b$.  Substitution of free variables and $\alpha$-conversion are deduced
(see~\cite{PPC} for details on \PPC, or Figures~\ref{FreeNames}
and~\ref{Substitution} for a formal definition in an extended setting).

\[
[\hskip1.5mm
  \begin{tikzpicture}[remember picture,overlay,baseline=(theta.base)]
    \node[inner sep=1pt](theta){$x$};
  \end{tikzpicture}
\hskip1.5mm]\hskip2mm x\hskip2mm
\begin{tikzpicture}[remember picture,overlay,baseline=(hatx.base)]
  \node[inner sep=0pt](hatx){$\hat{x}$};
\end{tikzpicture}
\hskip2mm\shortrightarrow\hskip2mm
\begin{tikzpicture}[remember picture,overlay,baseline=(x.base)]
  \node[inner sep=1pt](x){$x$}; \coordinate[above=1pt of hatx](hatxn);
  \draw[line width=1pt] (theta) ..controls +(0.2,0.25) and
  +(-0.2,0.4).. (hatxn); \draw[line width=1pt] (theta) ..controls +(0.3,0.55)
  and +(-0.2,0.5).. (x.north);
\end{tikzpicture}
\hskip2mm\hat{x} \quad=_\alpha\quad\case[\,y\,]{\ \ x\ \hat{y}\ }{\ y\ \hat{x}}
\]

One of the features of \PPC is the use of a single syntactic application for two
different meanings: the term $t_1t_2$ may represent either the usual {\bf functional
  application} of a function $t_1$ to an argument $t_2$  or the construction of a
data structure by {\bf structural application} of a constructor to one or more
arguments.  The latter is invariant: any structural application is forever a
data structure, whereas the functional application may be evaluated or
instantiated someday (and then turn into anything else, including a structural
application).

The simplest notion of pattern matching is syntactic: an argument $a$
matches a pattern $p$ if and only if there is a substitution $\sigma$ such that
$a=p^\sigma$. However, with arbitrary patterns, this solution generates
non-confluent calculi~\cite{LambdaPatOld}.  To recover confluence, syntactic matching
can be used together with a restriction on patterns, as for instance the
{\em rigid pattern condition} of the lambda-calculus with
patterns~\cite{LambdaPat}. The alternative
solution of \PPC allows a priori any term to be a pattern, and checks the
validity of patterns only a posteriori, when pattern matching is performed. In
particular, the restriction on patterns applies only once the evaluation of the
pattern is completed. This allows a greater freedom of evaluation and a greater
polymorphism of patterns, and hence a greater expressivity.

This is done by a more subtle notion of matching, called {\bf
  compound matching}, which tests whether patterns and arguments are in a
so-called {\bf matchable form}. A matchable form denotes a term which is
understood as a value, or in other words a term whose current form is stable and
then allows matching.
Matchable forms are described in \PPC at the meta-level by
the following grammar:
\begin{center}
\begin{tabular}{r@{$\quad::=\quad$}l@{\hspace{1.6cm}}l}
$d$ & $\hat{x}\ |\ dt$ & \PPC data structures\\
$m$ & $d\ |\ \case{t}{t}$ & \PPC matchable forms
\end{tabular}
\end{center}

\noindent
Compound matching is then defined (still at the meta-level) by the following
equations, taken in order.

\begin{center}
\begin{tabular}{r@{$\quad:=\quad$}ll}
$\cmatch{a}{\hat{x}}$ & $\set{x\mapsto a}$ & if $x\in\theta$\\
$\cmatch{\hat{x}}{\hat{x}}$ & $\set{}$ & if $x\not\in\theta$\\
$\cmatch{a_1a_2}{p_1p_2}$ & $\cmatch{a_1}{p_1}\uplus\cmatch{a_2}{p_2}$ & if
  $a_1a_2$ and $p_1p_2$ are matchable forms\\
$\cmatch{a}{p}$ & \fail & if $a$ and $p$ are matchable forms, otherwise\\
$\cmatch{a}{p}$ & \wait & otherwise
\end{tabular}
\end{center}

Its result, called a {\bf match} and denoted by $\rho$, may be a substitution
(written $\sigma$), a matching failure (written \fail) or the special value
\wait. The latter case represents undefined cases of matching, when the pattern
or the argument has still to be evaluated or instantiated before being matched.

Decomposition of compound patterns in the equations above is associated with an
operation $\uplus$ of disjoint union which ensures linearity of patterns: no
matching variable should be used twice in the same pattern, or confluence would
be broken~\cite{Klop}. Its formal definition is:
\begin{itemize}
\item $\uplus$ is commutative.
\item $\fail\uplus\rho=\fail$ for any $\rho$ (even \wait).
\item $\wait\uplus\rho=\wait$ for $\rho\neq\fail$.
\item $\sigma_1\uplus\sigma_2=\fail$ if domains of $\sigma_1$ and $\sigma_2$ overlap.
\item $\sigma_1\uplus\sigma_2$ is the union of $\sigma_1$ and $\sigma_2$ otherwise.
\end{itemize}

Finally, \PPC has to deal with a problem related to the dynamics of patterns: a
matching variable may be erased from a pattern during its evaluation. In this
case, no part of the argument would be bound to this matching variable and then no
term would be substituted to the corresponding variable. Hence free variables
would not be preserved, which would make reduction ill-defined (see
Example~\ref{MatchingExample}). This is avoided
in \PPC by a last (meta-level) test, called {\em check}: the result
$\match{a}{p}$ of the matching of $a$ against $p$ is defined as follows.

\begin{itemize}
\item if $\cmatch{a}{p}=\fail$ then $\match{a}{p}=\fail$.
\item if $\cmatch{a}{p}=\sigma$ with $dom(\sigma)\neq\theta$ then
  $\match{a}{p}=\fail$.
\item if $\cmatch{a}{p}=\sigma$ with $dom(\sigma)=\theta$ then
  $\match{a}{p}=\sigma$.
\end{itemize}

\noindent
Remark that $\match{a}{p}$ is not defined if $\cmatch{a}{p}=\wait$.\\
Finally, the reduction $\reppc$ of \PPC is defined by a unique
reduction rule (applied in any context):

\[
(\case{p}{b})a\quad\re_{\beta_m}\quad b^{\match{a}{p}}
\]

\noindent
where for any $b$ and $\sigma$ the expression $b^\sigma$ denotes the application of
the substitution $\sigma$ to the term $b$, and $b^\fail$ denotes some fixed
closed normal term $\fail$.

\begin{example}\label{MatchingExample}
Let $t$ be a \PPC term. The redex $(\case[x]{\hat{c}\hat{x}}{x})\ (\hat{c}t)$
reduces to $t$: the constructor $\hat{c}$ matches itself and the matchable
$\hat{x}$ is associated to $t$. On the other hand,
$(\case[x,y]{\hat{c}\hat{x}}{xy})\ (\hat{c}t)$ reduces to $\fail$: whereas the
compound matching is defined and successful, the check fails since there is no
match for $y$ and the result would be $ty$ where $y$ appears as a free
variable. The redex $(\case[x]{\hat{c}\hat{x}}{x})\ (\hat{c})$ also reduces to
$\fail$ since a constructor will never match a structural application. And last,
$(\case[x]{y\hat{x}}{x})\ (\hat{c}t)$ is not a redex since the pattern
$y\hat{x}$ has to be instantiated.
\end{example}

\section{Explicit Matching}
\label{EPPC}

This section defines the {\em Pure Pattern Calculus with Explicit
  Matching} (\EPPC), a calculus which gives an account of all the steps of a
pattern matching process of \PPC. The first point discussed is the
identification of structural application (Section~\ref{EDS}). An explicit
calculus is then fully detailed (Section~\ref{EPM}) and some of its basic
properties are proved (Section~\ref{Properties}). Explicit formulations of
simpler pattern calculi already appear in~\cite{Cut,Operators,Rho}.

\subsection{Explicit Data Structures}
\label{EDS}

Firstly, a new syntactic construct is introduced to discriminate between
functional and structural applications (as in~\cite{MatchingINets} for the
rewriting calculus for instance). Any application is supposed functional {\em a
  priori}, and two reduction rules propagate structural
information. The explicit structural application of $t$ to $u$ is written
$t\bullet u$.

\begin{center}
\begin{tabular}{r@{$\quad::=\quad$}l@{\hspace{1.6cm}}l}
$t$ & $x\ |\ \hat{x}\ |\ tt\ |\ t\bullet t\ |\ \case{t}{t}$ & \DSPPC terms\\
$d$ & $\hat{x}\ |\ t\bullet t$ & \DSPPC data structures
\end{tabular}
\end{center}

\begin{center}
\begin{tabular}{r@{$\quad\reb\quad$}l}
$\hat{x}\ t$ & $\hat{x}\bullet t$\\
$(t_1\bullet t_2)\ t_3$ & $(t_1\bullet t_2)\bullet t_3$ 
\end{tabular}
\end{center}

The identity morphism embeds \PPC into \DSPPC. The subset of \DSPPC defined by
\PPC is referred to as the set of pure terms. On the other hand, a ``forgetful''
morphism maps \DSPPC terms back to \PPC terms (or pure terms):
\begin{center}
\begin{tabular}{r@{$\quad:=\quad$}l}
$\trad{x}$ & $x$\\
$\trad{\hat{x}}$ & $\hat{x}$\\
$\trad{t_1t_2}$ & $\trad{t_1}\trad{t_2}$\\
$\trad{t_1\bullet t_2}$ & $\trad{t_1}\trad{t_2}$\\
$\trad{\case{p}{b}}$ & $\case{\trad{p}}{\trad{b}}$
\end{tabular}
\end{center}

Some \DSPPC data structures are not mapped to data structures of \PPC, for instance
$(\case{p}{b})\bullet a$. However, for any pure term $t$, if $t\reb^* t'$ and
$t'$ is a \DSPPC data structure, then $t$ is a \PPC data structure (proof by
induction on $t$).
One can also observe that for every \PPC data structure $t$, there exists a
reduction $t\reb^*t'$ with $t'$ a \DSPPC data structure.
Call {\bf well-formed} a term $t$ such that $\trad{t}\ \reb^*\ t$.

\subsection{Explicit Pattern Matching}
\label{EPM}

Another new syntactic object has to be introduced to represent an ongoing
matching operation. The basic information contained in such an object are: the
list of matching variables, a partial result recording what has already been
computed, and a representation of what has still to be solved.

This new object is called {\bf matching} and is written $\matching{\mu}{\Delta}$
with $\theta$ a list of names, $\mu$ a {\bf decided match} (that means, \fail or
a substitution), and $\Delta$ the collection of submatchings that have still to
be solved (a multiset of pairs of terms). For now on, we will consider
only decided matches, written $\mu$ (\wait does not exist as such in \EPPC).

The complete new grammar is:
\begin{center}
\begin{tabular}{r@{$\quad::=\quad$}l@{\hspace{1.6cm}}l}
$t$ & $x\ |\ \hat{x}\ |\ tt\ |\ t\bullet t\ |\ \case{t}{t}\ |\ t\matching{\mu}{\Delta}$ & \EPPC terms\\
$d$ & $\hat{x}\ |\ t\bullet t$ & \EPPC data structures\\
$m$ & $d\ |\ \case{t}{t}$ & \EPPC matchable forms
\end{tabular}
\end{center}

\begin{figure}[p]
The set of free names of a term $t$ is $\fn{t}=\fv{t}\cup\fm{t}$.

\[\begin{array}{r@{\quad:=\quad}l}
\multicolumn{2}{l}{\mathrm{Free\ variables}}\\
\fv{x}&\set{x}\\
\fv{\hat{x}}&\emptyset\\
\fv{t_1t_2}&\fv{t_1}\cup\fv{t_2}\\
\fv{t_1\bullet t_2}&\fv{t_1}\cup\fv{t_2}\\
\fv{\case{p}{b}}&\fv{p}\cup(\fv{b}\setminus\theta)\\
\fv{t\matching{\mu}{\Delta}}&(\fv{t}\setminus\theta)\cup\fv{codom(\mu)}\cup\fv{\Delta}\\
\multicolumn{2}{l}{}\\
\multicolumn{2}{l}{\mathrm{Free\ matchables}}\\
\fm{x}&\emptyset\\
\fm{\hat{x}}&\set{x}\\
\fm{t_1t_2}&\fm{t_1}\cup\fm{t_2}\\
\fm{t_1\bullet t_2}&\fm{t_1}\cup\fm{t_2}\\
\fm{\case{p}{b}}&(\fm{p}\setminus\theta)\cup\fm{b}\\
\fm{t\matching{\mu}{\Delta}}&\fm{t}\cup\fm{codom(\mu)}\cup\fm{\pi_1(\Delta)}\cup(\fm{\pi_2(\Delta)}\setminus\theta)
\end{array}\]

\noindent
where if $\Delta=(a_1,p_1)...(a_n,p_n)$ then
$\fm{\pi_1(\Delta)}=\bigcup_i\fm{a_i}$ and
$\fm{\pi_2(\Delta)}=\bigcup_i\fm{p_i}$.

\caption{Free names of a \EPPC term}
\label{FreeNames}
\end{figure}

\begin{figure}[p]
\[\begin{array}{r@{\quad:=\quad}l@{\qquad}l}
x^\sigma & \sigma_x & x\in dom(\sigma)\\
x^\sigma & x & x\not\in dom(\sigma)\\
\hat{x}^\sigma & \hat{x}\\
(tu)^\sigma & t^\sigma u^\sigma\\
(t\bullet u)^\sigma & t^\sigma \bullet u^\sigma\\
(\case{p}{b})^\sigma & (\case{p^\sigma}{b^\sigma}) & \theta\cap(dom(\sigma)\cup\fn{\sigma})=\emptyset\\
(t\matching{\mu}{\Delta})^\sigma & t^\sigma\matching{\mu^\sigma}{\Delta^\sigma} &
  \theta\cap(dom(\sigma)\cup\fn{\sigma})=\emptyset
\end{array}\]

\noindent
where in $\Delta^\sigma$ (resp. $\mu^\sigma$) the substitution propagates in all
terms of $\Delta$ (resp. of the codomain of $\mu$).

\caption{Substitution in \EPPC}
\label{Substitution}
\end{figure}

\begin{figure}[p]

\noindent
\textbf{Initialisation}

\begin{center}
\begin{tabular}{r@{$\quad\reB\quad$}ll}
$(\case{p}{b})a$ & $b\matching{\emptyset}{(a,p)}$
\end{tabular}
\end{center}

\noindent
\textbf{Structural application}

\begin{center}
\begin{tabular}{r@{$\quad\reb\quad$}l}
$\hat{x}\ t$ & $\hat{x}\bullet t$\\
$(t_1\bullet t_2)\ t_3$ & $(t_1\bullet t_2)\bullet t_3$ 
\end{tabular}
\end{center}

\noindent
\textbf{Matching}\\
Since $\Delta$ has been defined as a multiset of pairs of terms, its elements
are not ordered. In the following rules $(a,p)\Delta$ denotes the (multiset)
union of $\Delta$ with the singleton $\set{(a,p)}$.\bigskip

\noindent
The first three matching rules are for successful matching steps.
\begin{center}
\begin{tabular}{r@{$\quad\rem\quad$}ll}
  $b\matching{\mu}{(a,\hat{x})\Delta}$ & $b\matching{\mu\uplus\set{x\mapsto a}}{\Delta}$ & if $x\in\theta$ and $\fn{a}\cap\theta=\emptyset$\\
$b\matching{\mu}{(\hat{x},\hat{x})\Delta}$ & $b\matching{\mu}{\Delta}$ & if
  $x\not\in\theta$\\
  $b\matching{\mu}{(a_1\bullet a_2,p_1\bullet p_2)\Delta}$ &
  $b\matching{\mu}{(a_1,p_1)(a_2,p_2)\Delta}$
\end{tabular}
\end{center}\bigskip

\noindent
The last six matching rules are for failure, and could be summed up as ``for any
other matchable forms $a$ and $p$, let $b\matching{\mu}{(a,p)\Delta}$ reduce to
$b\matching{\fail}{\Delta}$''.
\begin{center}
\begin{tabular}{r@{$\quad\rem\quad$}ll}
$b\matching{\mu}{(\hat{y},\hat{x})\Delta}$ & $b\matching{\fail}{\Delta}$ & if $x\not\in\theta$ and $x\neq y$\\
$b\matching{\mu}{(a_1\bullet a_2,\hat{x})\Delta}$ & $b\matching{\fail}{\Delta}$ & if $x\not\in\theta$\\
$b\matching{\mu}{(\case[\theta_a]{p_a}{b_a},\hat{x})\Delta}$ & $b\matching{\fail}{\Delta}$ & if $x\not\in\theta$\\
$b\matching{\mu}{(\hat{x},p_1\bullet p_2)\Delta}$ & $b\matching{\fail}{\Delta}$ \\
$b\matching{\mu}{(\case[\theta_a]{p_a}{b_a},p_1\bullet p_2)\Delta}$ & $b\matching{\fail}{\Delta}$ \\
$b\matching{\mu}{(a,\case[\theta_p]{p_p}{b_p})\Delta}$ & $b\matching{\fail}{\Delta}$
\end{tabular}
\end{center}\medskip

\noindent
\textbf{Resolution}

\begin{center}
\begin{tabular}{r@{$\quad\rer\quad$}lll}
$b\matching[\theta]{\sigma}{\emptyset}$ & $b^\sigma$ & if $dom(\sigma)=\theta$ &
  ({\bf substitution rule})\\
$b\matching[\theta]{\sigma}{\emptyset}$ & \fail & if $dom(\sigma)\neq\theta$\\
$b\matching[\theta]{\fail}{\Delta}$ & \fail
\end{tabular}
\end{center}

\caption{Rules of \EPPC}
\label{EPPCRules}
\end{figure}

A pure term of \EPPC is a term without any structural application or matching
(that means a \PPC term). As in \PPC, the symbol \fail used as a term denotes a
fixed closed pure normal term.

Free variables and matchables are defined in Figure~\ref{FreeNames} as a natural
extension of \PPC mechanisms to explicit matching.
Similarly, a notion of (meta-level) substitution is deduced from this definition
(Figure~\ref{Substitution}).
Finally, a notion of $\alpha$-conversion is associated, and from now, on it is
supposed that all bound names in a term are different, and disjoint from free
names.

New rules for matching are of three kinds: an {\em initialisation rule} $\reB$
which triggers a new matching operation, several {\em matching rules} $\rem$
corresponding to all possible elementary matching steps and three {\em resolution
  rules} $\rer$ that apply the result of a completed matching. The complete set
of rules of \EPPC is given in Figure~\ref{EPPCRules}.

Reduction $\reeppc$ of \EPPC is defined by application of any rule of $\reB$,
$\reb$, $\rem$ or $\rer$ in any context.  The subsystem
$\rep\ =\ \reb\cup\rem\cup\rer$ computes (when possible) already existing
pattern matchings but does not create new ones. 

\subsection{Confluence and Simulation properties}
\label{Properties}

This section states and proves four theorems on basic properties of \EPPC and
its links with \PPC. The first one is a
result on the normalization of already existing pattern matchings.

\begin{theorem}
$\rep$ is confluent and strongly normalizing.
\end{theorem}

\begin{proof}\

\begin{itemize}
\item We define two well-founded orders \nestor and \edward, whose lexicographic
  product contains ${}_p\!\!\rre$. This will enforce strong normalization.
\begin{itemize}
\item \nestor sorts terms with respect to the nesting of matchings. It is
  based on an over-approximation of the depth of potentially nested matchings
  (matchings that are syntactically nested or that may become such after some
  substitutions).
  For any lists of names $\theta_i$, decided matches $\mu_i$, and lists of pairs of
  terms $\Delta_i$, the sequence
  $\matching[\theta_1]{\mu_1}{\Delta_1};...;\matching[\theta_n]{\mu_n}{\Delta_n}$
  is called a potentially nested chain of length $n$ if for each
  $i\in\set{1...n-1}$ one of these conditions holds:
  \begin{itemize}
  \item {\bf Nesting:} $\matching[\theta_{i+1}]{\mu_{i+1}}{\Delta_{i+1}}$
    appears in $\Delta_i$ or in the codomain of $\mu_i$.
  \item {\bf Potential nesting:} a variable of $\theta_{i+1}$ appears in
    $\Delta_i$ or in the codomain of $\mu_i$.
  \end{itemize}
  The set of maximal chains of a term $t$ is the set of all
  potentially nested chains that can be built using the matchings appearing in
  $t$ and that can not be extended (neither by the left nor by the right) using
  other matchings of $t$. For this extraction, remember that all bound names in
  $t$ are supposed to be different, and disjoint from free names. The depth of
  $t$ is the multiset of the lengths of the maximal chains of $t$.
  \begin{example}
    Write
    $t=\hat{c}\,\matching[\rule{0pt}{10pt}\emptyset\,]{\,\emptyset\,}{\,(x,\hat{c})(x,\hat{c})}\,\matching[\rule{0pt}{10pt}\,x\,]{\,x\mapsto
        y\matching[y]{\emptyset}{(\hat{c},\hat{y})}\,}{\,\emptyset\,}$.
    The term $t$ contains three matchings and has one maximal chain of length
    $3$, which is
    \[\matching[\rule{0pt}{10pt}\,\emptyset\,]{\,\emptyset\,}{\,(x,\hat{c})(x,\hat{c})\,}\,;\,\matching[\rule{0pt}{10pt}\,x\,]{\,x\mapsto y\matching[y]{\emptyset}{(\hat{c},\hat{y})}\,}{\,\emptyset\,}\,;\,\matching[\rule{0pt}{10pt}\,y\,]{\,\emptyset\,}{\,(\hat{c},\hat{y})\,}\]
    The reduction
    $t\ \rer\ t'=\hat{c}\,\matching[\rule{0pt}{10pt}\,\emptyset\,]{\,\emptyset\,}{\,(y_1\matching[y_1]{\emptyset}{(\hat{c},\hat{y_1})},\hat{c})\,\,(y_2\matching[y_2]{\emptyset}{(\hat{c},\hat{y_2})},\hat{c})\,}$
    yields a new term $t'$ which still contains three matchings (one was reduced
    and disappeared but another one was duplicated) and admits two maximal
    chains of length $2$, namely
    \[\matching[\rule{0pt}{10pt}\,\emptyset\,]{\,\emptyset\,}{\,(y_1\matching[y_1]{\emptyset}{(\hat{c},\hat{y_1})},\hat{c})\,\,(y_2\matching[y_2]{\emptyset}{(\hat{c},\hat{y_2})},\hat{c})\,}\,;\,\matching[\rule{0pt}{10pt}\,y_1\,]{\,\emptyset\,}{\,(\hat{c},\hat{y_1})\,}\]\[\matching[\rule{0pt}{10pt}\,\emptyset\,]{\,\emptyset\,}{\,(y_1\matching[y_1]{\emptyset}{(\hat{c},\hat{y_1})},\hat{c})\,\,(y_2\matching[y_2]{\emptyset}{(\hat{c},\hat{y_2})},\hat{c})\,}\,;\,\matching[\rule{0pt}{10pt}\,y_2\,]{\,\emptyset\,}{\,(\hat{c},\hat{y_2})\,}\]
  \end{example}
  The usual order on natural integers gives a well-founded order on the lengths
  of potentially nested chains. \nestor is defined as the multiset extension of
  this order, applied to the depths of terms. It strictly decreases for any
  reduction by the substitution rule, and is less or equal for any other
  reduction.
\item \edward is the natural order on the size of terms, defined as follows:
  \[\begin{array}{r@{\quad:=\quad}l}
  \size{x} & 1\\
  \size{\hat{x}} & 1\\
  \size{t_1t_2} & \size{t_1}+\size{t_2}+2\\
  \size{t_1\bullet t_2} & \size{t_1}+\size{t_2}+1\\
  \size{\case{p}{b}} & \size{p}+\size{b}\\
  \size{b\matching{\mu}{\Delta}} & \size{b} + \size{\fail} +
  \sum_{x\in dom(\mu)}\size{\mu_x} + \sum_{(a,p)\in_k\Delta}k(\size{a}+\size{p})
  \end{array}\]
  where we write $e\in_k\Delta$ when the element $e$ appears in the multiset
  $\Delta$ with multiplicity $k$.

  \edward strictly decreases for any reduction except by the substitution rule.
\end{itemize}
\item Matching rules generate some critical pairs, most of which are trivially
  convergent. The most subtle case is the reduction of a non linear matching:

\begin{center}
$\matching{\mu\uplus\set{x\mapsto a_1}}{(a_2,\hat{x})\Delta}
\quad_{p}\reflectbox{$\re$}\quad
\matching{\mu}{(a_1,\hat{x})(a_2,\hat{x})\Delta}
\quad\rep\quad
\matching{\mu\uplus\set{x\mapsto a_2}}{(a_1,\hat{x})\Delta}$
\end{center}

Since $\uplus$ is a disjoint union of substitutions, both sides can be reduced
to $\matching{\fail}{\Delta}$.\\ Finally, $\rep$ is weakly confluent, and then
confluent by Newman's Lemma~\cite{Terese}.
\end{itemize}

\end{proof}

The second theorem states the confluence of $\reeppc$. Since the reduction of
\EPPC is defined by several rules, the result does not fall into the modular
framework of~\cite{PPC}. It is proved here directly by the Tait and
Martin-Löf's technique. The main construction of the proof is the definition (in
Figure~\ref{ParallelReduction}) of a parallel reduction relation $\Pre$ enjoying
the diamond property (Lemma~\ref{Diamond}). The relation $\Pre$ is first linked
to $\reeppc$ in Lemma~\ref{ParallelSimulation}.

\begin{figure}
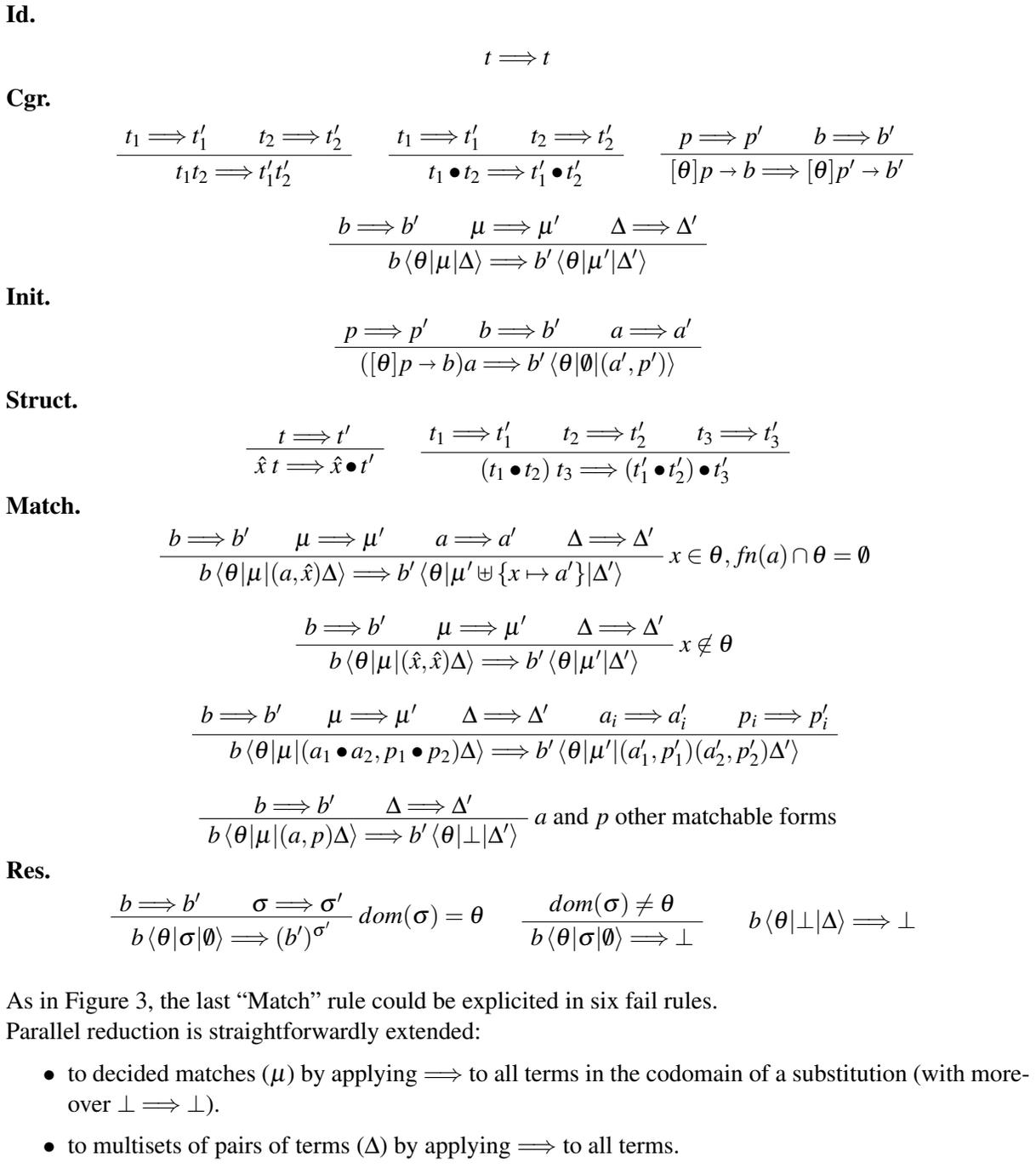


\noindent
\textbf{Id.}
\begin{center}
\AxiomC{$t\Pre t$}
\DisplayProof
\end{center}

\noindent
\textbf{Cgr.}
\begin{center}
\AxiomC{$t_1\Pre t'_1$}
\AxiomC{$t_2\Pre t'_2$}
\BinaryInfC{$t_1t_2\Pre t'_1t'_2$}
\DisplayProof\quad
\AxiomC{$t_1\Pre t'_1$}
\AxiomC{$t_2\Pre t'_2$}
\BinaryInfC{$t_1\bullet t_2\Pre t'_1\bullet t'_2$}
\DisplayProof\quad
\AxiomC{$p\Pre p'$}
\AxiomC{$b\Pre b'$}
\BinaryInfC{$\case{p}{b}\Pre\case{p'}{b'}$}
\DisplayProof\bigskip\\
\AxiomC{$b\Pre b'$}
\AxiomC{$\mu\Pre\mu'$}
\AxiomC{$\Delta\Pre\Delta'$}
\TrinaryInfC{$b\matching{\mu}{\Delta}\Pre b'\matching{\mu'}{\Delta'}$}
\DisplayProof
\end{center}

\noindent
\textbf{Init.}
\begin{center}
\AxiomC{$p\Pre p'$}
\AxiomC{$b\Pre b'$}
\AxiomC{$a\Pre a'$}
\TrinaryInfC{$(\case{p}{b})a\Pre b'\matching{\emptyset}{(a',p')}$}
\DisplayProof
\end{center}

\noindent
\textbf{Struct.}
\begin{center}
\AxiomC{$t\Pre t'$}
\UnaryInfC{$\hat{x}\ t\Pre \hat{x}\bullet t'$}
\DisplayProof\quad
\AxiomC{$t_1\Pre t'_1$}
\AxiomC{$t_2\Pre t'_2$}
\AxiomC{$t_3\Pre t'_3$}
\TrinaryInfC{$(t_1\bullet t_2)\ t_3\Pre (t'_1\bullet t'_2)\bullet t'_3$}
\DisplayProof
\end{center}

\noindent
\textbf{Match.}
\begin{center}
\AxiomC{$b\Pre b'\hspace{7mm}\mu\Pre\mu'$}
\AxiomC{$a\Pre a'$}
\AxiomC{$\Delta\Pre\Delta'$}
\RightLabel{$x\in\theta, \fn{a}\cap\theta=\emptyset$}
\TrinaryInfC{$b\matching{\mu}{(a,\hat{x})\Delta}\Pre
b'\matching{\mu'\uplus\set{x\mapsto a'}}{\Delta'}$}
\DisplayProof\bigskip\\
\AxiomC{$b\Pre b'$}
\AxiomC{$\mu\Pre\mu'$}
\AxiomC{$\Delta\Pre\Delta'$}
\RightLabel{$x\not\in\theta$}
\TrinaryInfC{$b\matching{\mu}{(\hat{x},\hat{x})\Delta}\Pre
  b'\matching{\mu'}{\Delta'}$}
\DisplayProof\bigskip\\
\AxiomC{$b\Pre b'\hspace{7mm}\mu\Pre\mu'\hspace{7mm}\Delta\Pre\Delta'$}
\AxiomC{$a_i\Pre a'_i$}
\AxiomC{$p_i\Pre p'_i$}
\TrinaryInfC{$b\matching{\mu}{(a_1\bullet a_2,p_1\bullet p_2)\Delta}\Pre 
b'\matching{\mu'}{(a'_1,p'_1)(a'_2,p'_2)\Delta'}$}
\DisplayProof\bigskip\\
\AxiomC{$b\Pre b'$}
\AxiomC{$\Delta\Pre\Delta'$}
\RightLabel{$a$ and $p$ other matchable forms}
\BinaryInfC{$b\matching{\mu}{(a,p)\Delta}\Pre b'\matching{\fail}{\Delta'}$}
\DisplayProof
\end{center}

\noindent
\textbf{Res.}
\begin{center}
\AxiomC{$b\Pre b'$}
\AxiomC{$\sigma\Pre\sigma'$}
\RightLabel{$dom(\sigma)=\theta$}
\BinaryInfC{$b\matching{\sigma}{\emptyset}\Pre (b')^{\sigma'}$}
\DisplayProof\quad
\AxiomC{$dom(\sigma)\neq\theta$}
\UnaryInfC{$b\matching{\sigma}{\emptyset}\Pre\fail$}
\DisplayProof\quad
\AxiomC{$b\matching{\fail}{\Delta}\Pre\fail$}
\DisplayProof
\end{center}\bigskip

As in Figure~\ref{EPPCRules}, the last ``Match'' rule could be explicited
in six fail rules.

Parallel reduction is straightforwardly extended:
\begin{itemize}
\item to decided matches ($\mu$) by applying $\Pre$ to all terms in the codomain of a
  substitution (with moreover $\fail\Pre\fail$).
\item to multisets of pairs of terms ($\Delta$) by applying $\Pre$ to all terms.
\end{itemize}

\caption{Definition of parallel reduction relation $\Pre$}
\label{ParallelReduction}
\end{figure}

\begin{lemma}\label{ParallelSimulation}
$\reeppc\quad\subseteq\quad\Pre\quad\subseteq\quad\reeppc^*$
\end{lemma}
\begin{proof}\
\begin{itemize}
\item $\reeppc\quad\subseteq\quad\Pre$ \ by induction on the definition of $\reeppc$.
\item $\Pre\quad\subseteq\quad\reeppc^*$ \ by induction on the definition of $\Pre$.
\end{itemize}
\end{proof}

\begin{lemma}\label{ReductionSubstitution}
If $t\Pre t'$ and $\sigma\Pre\sigma'$ then $t^\sigma\Pre t'^{\sigma'}$.
\end{lemma}
\begin{proof}
By induction on the derivation of $t\Pre t'$.
\end{proof}

\begin{lemma}\label{Diamond}
$\rPre\Pre\quad\subseteq\quad\Pre\rPre$
\end{lemma}
\begin{proof}
Suppose $t_1\rPre t\Pre t_2$. Induction on
the derivations of $t\Pre t_1$ and $t\Pre t_2$:
\begin{itemize}
\item If one of the reductions is by ``Id'', the conclusion is immediate.
\item If one reduction is by a ``Cgr'' rule, and the other by a ``Cgr'', ``Init'',
  ``Struct'', or ``Match'' rule, then the induction hypothesis applies
  straightforwardly.
\item If one reduction is by a ``Cgr'' rule and the other by a ``Res'' rule,
  there is one non trivial case: suppose $t_1\matching{\sigma_1}{\emptyset}\rPre
  t\matching{\sigma}{\emptyset}\Pre t_2^{\sigma_2}$. By induction hypothesis
  there are $t_3$ and $\sigma_3$ such that
  $t_1\Pre t_3\rPre t_2$ and $\sigma_1\Pre\sigma_3\rPre\sigma_2$. Then we can
  derive $t_1\matching{\sigma_1}{\emptyset}\Pre t_3^{\sigma_3}$. Finally, by
  Lemma~\ref{ReductionSubstitution} we conclude that $t_2^{\sigma_2}\Pre
  t_3^{\sigma_3}$.
\item If both reductions are by a ``Init'' rule, then the induction hypotheses
  apply straightforwardly.
\item Idem for two ``Struct'' or two ``Match'' rules.
\item Case where both reductions are by a ``Res'' rule. Reductions to $\fail$ are
  straightforward. Then consider the following case: $t_1^{\sigma_1}\rPre
  t\matching{\sigma}{\emptyset}\Pre t_2^{\sigma_2}$. By induction hypotheses
  $t_1\Pre t_3\rPre t_2$ and $\sigma_1\Pre\sigma_3\rPre\sigma_2$. By
  Lemma~\ref{ReductionSubstitution} $t_1^{\sigma_1}\Pre t_3^{\sigma_3}\rPre t_2^{\sigma_2}$.
\end{itemize}
\end{proof}

\begin{theorem}
\EPPC is confluent.
\end{theorem}

\begin{proof}
Since $\Pre$ has the diamond property (Lemma~\ref{Diamond}), its transitive
closure $\Pre^*$ also enjoys the diamond property (\cite{Terese}).
Moreover Lemma~\ref{ParallelSimulation} implies $\reeppc^*\ =\ \Pre^*$, and then
$\reeppc^*$ enjoys the diamond property. Finally, $\reeppc$ is confluent.
\end{proof}

\noindent
The last two theorems establish a link between the calculus with explicit
matching \EPPC and the original implicit \PPC.

\begin{lemma}\label{MatchSimulation}
If $\cmatch{a}{p}=\mu$ with $\mu$ a decided match, then for any $\mu_0$ and
$\Delta$ there are $\mu'$ with $\trad{\mu'}=\mu$ and a reduction
\[\matching{\mu_0}{(a,p)\Delta}\ (\reb\cup\rem)^*\ \matching{\mu_0\uplus\mu'}{\Delta}\]
\end{lemma}

\begin{proof}
Induction on $\cmatch{a}{p}$.

\begin{itemize}
\item $\cmatch{a}{\hat{x}}$ with $x\in\theta$ or $\cmatch{\hat{x}}{\hat{x}}$
  with $x\not\in\theta$: immediate.
\item $\cmatch{aa_0}{pp_0}$ with $aa_0$ and $pp_0$ matchable forms. Hence
  $a=a_n...a_1$ and $p=p_m...p_1$ with $a_n$ and $p_m$ constructors. Then
  $a_n...a_1a_0\reb^*a_n\bullet...\bullet a_1\bullet a_0$ and
  $p_m...p_1p_0\reb^*p_m\bullet...\bullet p_1\bullet p_0$.
  Suppose $n\geq m$, then
  $\cmatch{aa_0}{pp_0}=\cmatch{a_m...a_n}{p_n}\uplus\cmatch{a_{n-1}}{p_{n-1}}\uplus...\uplus\cmatch{a_0}{p_0}$
  and $\matching{\mu_0}{(a_n\bullet...\bullet a_0,p_m\bullet...\bullet
    p_0)\Delta}\rem^*\matching{\mu_0}{(a_m\bullet...\bullet a_n,p_n)(a_{n-1},p_{n-1})...(a_0,p_0)\Delta}$.
  Case on $p_n=\hat{x}$:

  \begin{itemize}
  \item If $x\in\theta$ then the matching reduces to
  $\matching{\mu_0\uplus\set{x\mapsto a_m\bullet...\bullet
      a_n}}{(a_{n-1},p_{n-1})...(a_0,p_0)\Delta}$.
  \item If $x\not\in\theta$ then the matching reduces to
    $\matching{\mu'_0}{(a_{n-1},p_{n-1})...(a_0,p_0)\Delta}$ with $\mu_0'=\mu_0$
    or $\mu_0'=\fail$.
  \end{itemize}

  \noindent
  In any of these two cases, the induction hypothesis gives the conclusion.
  In the case where $m>n$, the same method allows to derive a reduction to \fail.
\item Cases of matching failure: for instance $\cmatch{\hat{x}}{\hat{y}t}$. The
  following reduction gives the conclusion:
  $\matching{\mu_0}{(\hat{x},\hat{y}t)\Delta}\reb\matching{\mu_0}{(\hat{x},\hat{y}\bullet
    t)\Delta}\rem \matching{\fail}{\Delta}$.
\end{itemize}

\end{proof}

\begin{theorem}
For any terms $t$ and $t'$ of \PPC, if \ $t\reppc t'$ \ then \ $t\reeppc^*t'$.
\end{theorem}

\begin{proof}
Suppose $t\reppc t'$. There is a context $C[]$ such that
$t=C[(\case{p}{b})a]\ \reppc\ C[b']=t'$ and $\cmatch{a}{p}=\mu$ with $\mu$ a
decided match.\\
By Lemma~\ref{MatchSimulation}
$(\case{p}{b})a\ \reB\  b\matching{\emptyset}{(a,p)}\ (\reb\cup\rem)^*\ b\matching{\mu}{\emptyset}$.\newpage

\noindent
Case on $\mu$:
\begin{itemize}
\item If $\mu=\fail$ then $b'=\fail$ and
  $b\matching{\fail}{\emptyset}\rer\fail$.
\item Else $\mu=\sigma$ and:
\begin{itemize}
\item If $dom(\sigma)=\theta$ then $b'=b^\sigma$ and
  $b\matching{\sigma}{\emptyset}\ \rer\ b^\sigma$.
\item Else $b'=\fail$ and $b\matching{\fail}{\emptyset}\ \rer\ \fail$.
\end{itemize}
\end{itemize}
\end{proof}

\noindent
The map $\trad{\cdot}$ is naturally extended to any \EPPC term, set of \EPPC
terms and decided match, as well as the notion of well-formedness.
Then, for any $\mu$ and $\Delta$ not containing any explicit matching, define the
semantics of the matching $\matching{\mu}{\Delta}$ by:
\[\sem{\mu}{\Delta}=\trad{\mu}\uplus\left(\biguplus_{(a,p)\in\Delta}\cmatch{\trad{a}}{\trad{p}}\right)\]
Note that the semantics can be \wait.

\begin{lemma}\label{SemanticsStability}
For any well-formed $\mu$, $\mu'$, $\Delta$ and $\Delta'$ which do not contain
any explicit matching,\\ if $\matching{\mu}{\Delta}\rem\matching{\mu'}{\Delta'}$
or $\matching{\mu}{\Delta}\reb\matching{\mu'}{\Delta'}$ then
$\sem{\mu}{\Delta}=\sem{\mu'}{\Delta'}$.
\end{lemma}
\begin{proof}
Case on the reduction rules.
\end{proof}

\begin{lemma}[\cite{PPC}]\label{PPCSubstitution}
If $t\reppc t'$, then $t^\sigma\reppc t'^\sigma$.
\end{lemma}

\noindent
Let $t$ be a \EPPC term, and $t'$ the unique normal form of $t$ by $\rep$. Write
$t\!\!\downarrow$ and call purification of $t$ the term $\trad{t'}$. Note that
the purification may not be a pure term if there is an unsolvable matching in it.

\begin{theorem}
For any well-formed terms $t$ and $t'$ of \EPPC,\\ if $t\reeppc t'$ and
$t\!\!\downarrow$ and $t'\!\!\downarrow$ are pure, then
$t\!\!\downarrow\,=t'\!\!\downarrow$\, or
$t\!\!\downarrow\ \reppc t'\!\!\downarrow$\,.
\end{theorem}

\begin{proof}
Induction on $t\reeppc t'$.

\begin{itemize}
\item Case $t=(\case{p}{b})a\ \reB\ b\matching{\emptyset}{(p,a)}=t'$.
The term $t'\pure$ is pure, then there is a sequence
$b\pure\!\matching{\emptyset}{(p\pure,a\pure)}\ (\reb\cup\rem)^*\ b\!\pure\matching{\mu}{\Delta}\rer
t''$ where
$\trad{t''}=t'\pure$ and where $\Delta=\emptyset$ or $\mu=\fail$. By
Lemma~\ref{SemanticsStability}, $\trad{\mu}=\cmatch{a\pure}{p\pure}$. Then, by
case on matching resolution, $t\pure\reppc\trad{t''}=t'\pure$.
\item Other base cases: if $t\rep t'$, then $t\!\!\downarrow\ = t'\!\!\downarrow$\,.
\item Case $t=b\matching{\mu}{\Delta}\reeppc b'\matching{\mu}{\Delta}=t'$.
The term $t\pure$ is pure. Then
$\matching{\mu}{\Delta}\rep^*\matching{\mu'}{\Delta'}$ where $\Delta'=\emptyset$
or $\mu'=\fail$. If $\mu'=\fail$ or $dom(\mu')\neq\theta$, then
$t\pure=t'\pure=\fail$. Suppose $\Delta'=\emptyset$ and $\mu'=\sigma$ with
$dom(\sigma)=\theta$. Hence $t\pure=(b\pure)^\sigma$ and
$t'\pure=(b'\pure)^\sigma$.
By induction hypothesis $b\pure\reppc b'\pure$, and then by
Lemma~\ref{PPCSubstitution} $t\pure\reppc t'\pure$.
\item Other inductive cases are straightforward.
\end{itemize}
\end{proof}

This section introduced the new calculus \EPPC for explicit matching
with dynamic patterns, and proved its confluence. It also expressed a
bidirectional simulation between \PPC and \EPPC: first any reduction of \PPC is
reflected in \EPPC by a sequence. On the other hand, a reduction of \EPPC can be
mapped on zero or one step of \PPC if and only if its source and its target are
well-formed and can be purified. Next section discusses how this new calculus
can be used.

\section{Discussion}

\subsection{Reduction Strategies}
\label{Applications}

Pattern matching raises at least two new issues concerning reduction strategies
({\em i.e.} the evaluation order of programs). One is related to the order in which
pattern matching steps are performed, the other concerns the amount of
evaluation of the pattern and of the argument performed before pattern matching is
solved.\vspace{-2mm}

\paragraph{Some remarks about the order of pattern matching steps.}~\\
\EPPC uses a multiset as the third component of a matching
$\matching{\mu}{\Delta}$ to represent all the remaining work. The 
calculus is thus able to cover all the possible orders of pattern matching
steps. A particular strategy may be enforced by giving more structure to
the multiset $\Delta$ and by adapting the matching reduction rules.

\begin{example}\label{ExampleList}
Suppose that $\Delta$ is now a {\em list} of pairs of terms, and $(a,p)\Delta$
denotes the usual ``cons'': it builds the list whose head is $(a,p)$ and
whose tail is $\Delta$. Then the rules of Figure~\ref{EPPCRules} implement a
depth-first, left-to-right pattern matching algorithm.
\end{example}

\begin{example}
Now assume the list structure of Example~\ref{ExampleList} and replace the right
member of the reduction rule
$\matching{\mu}{(a_1\bullet a_2,p_1\bullet p_2)\Delta}\ \rem\ \matching{\mu}{(a_1,p_1)(a_2,p_2)\Delta}$
by $\matching{\mu}{\Delta(a_1,p_1)(a_2,p_2)}$. Then pattern matching is done in
a completely different order!
\end{example}

\noindent
More generally, if some permutations of the elements of $\Delta$ are allowed, lots of
richer matching behaviours may be described in \EPPC.\vspace{-2mm}

\paragraph{Pattern and argument evaluation: what is needed?}~\\
In \PPC, a naive evaluation strategy for a term $(\case{p}{b})a$ could be:
evaluate the pattern $p$ and the argument $a$, then solve the matching
(atomically). As the usual call-by-value, this solution may perform unneeded
evaluation of the argument, for instance in parts that are not reused in the
body $b$ of the function. The most basic solution to this problem, call-by-name,
allows the substitution of non-evaluated arguments. But how can such a solution
be described in a pattern calculus?

In the context of pattern matching, some evaluation of the argument has to be
done before pattern matching is solved. However the exact amount of needed
evaluation depends on the pattern. Hence pattern matching enforces some
kind of call-by-value where the notion of value is context-sensitive.
Moreover, even the evaluation of the pattern may depend on the argument!

This makes the description of a strategy performing a minimal evaluation of the
dynamic pattern and the argument rather difficult. One may keep for the
object-level a compact formalism like \PPC by defining complex meta-level
operations finely parametrised by terms. This is done in~\cite{Standard} to
describe standard reductions in a simpler pattern calculus. In contrast to this
solution, we want to show here how the richer syntax of \EPPC allows a simple
description of such a reduction strategy.

Indeed \EPPC allows to interleave pattern and argument reduction with
pattern matching steps. This finer control allows for instance an easy
definition of a ``matching-driven'' reduction, as pictured in
Figure~\ref{StrategyPicture}.
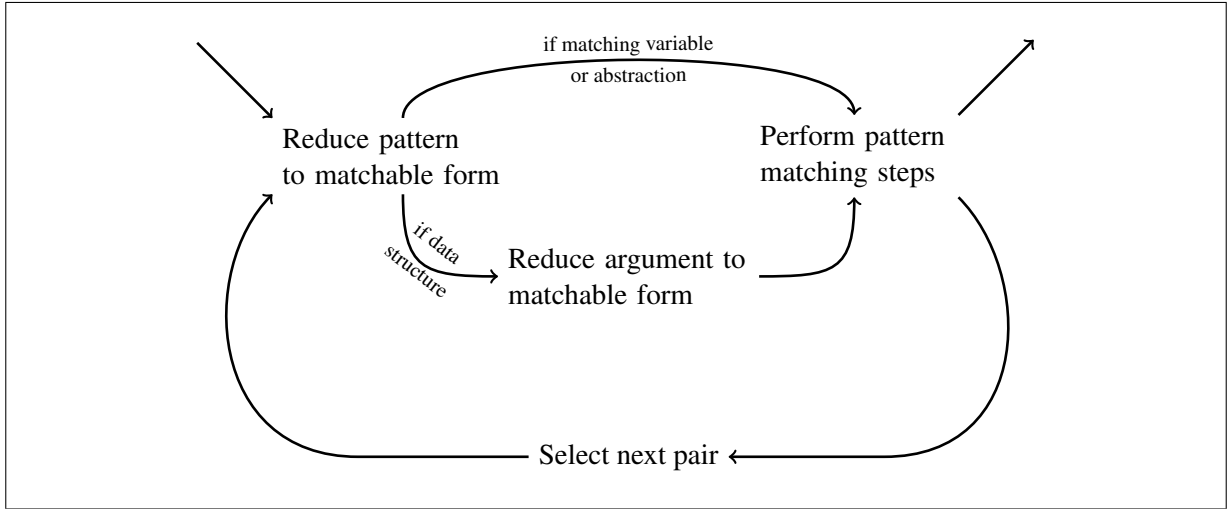
\begin{figure}
\begin{center}
\vskip2mm
\begin{tikzpicture}
\node[anchor=center,text width=3.2cm](pat){Reduce pattern\\ to matchable form};
\node[anchor=center,text width=3.2cm](arg) at (3,-1.6) {Reduce argument to matchable form};
\node[anchor=center,text width=2.5cm](mat) at (6,0) {Perform pattern matching steps};
\node[anchor=center] (nex) at (3,-4) {Select next pair};
\draw[->,line width=1pt] (pat.north) ..controls +(0,1) and
+(0,1).. (mat.north) node[sloped,pos=0.5,above=-0.8mm]{\scriptsize if matching variable} node[sloped,pos=0.5,below=-0.5mm]{\scriptsize or abstraction};
\draw[->,line width=1pt] (pat.south) ..controls +(0,-1) and
+(-1,0).. (arg.west) node[anchor=center,sloped,pos=0.5,above=1mm]{\scriptsize if data}
node[anchor=center,sloped,pos=0.5,below=-1mm]{\scriptsize structure};
\draw[->,line width=1pt] (arg.east) ..controls +(1,0) and +(0,-1).. (mat.south);
\draw[->,line width=1pt] (mat.south east) ..controls +(1,-1) and +(2,0).. (6.4,-4)--(nex.east);
\draw[<-,line width=1pt] (pat.south west) ..controls +(-1,-1) and +(-2,0).. (-0.6,-4)--(nex.west);
\draw[<-,line width=1pt] (pat.north west)--+(-1,1);
\draw[->,line width=1pt] (mat.north east)--+(1,1);
\end{tikzpicture}
\end{center}

\caption{Matching-driven reduction strategy}
\label{StrategyPicture}
\end{figure}

The idea here is to trigger pattern matchings as soon
as possible. Then the pattern and the argument are evaluated until they become
matchable, and one or more pattern matching steps are performed before the story
goes on. A formal definition of a strategy implementing this picture is
by restricting the reduction under a context to the only four rules given in
Figure~\ref{Contexts}.

\begin{figure}
\begin{center}
\AxiomC{$t_1\ \re\ t'_1$}
\UnaryInfC{$t_1t_2\ \re\ t'_1t_2$}
\DisplayProof\bigskip\\
\AxiomC{$p\ \re\ p'$}
\UnaryInfC{$b\matching{\mu}{(a,p)\Delta}\ \re\ b\matching{\mu}{(a,p')\Delta}$}
\DisplayProof\bigskip\\
\AxiomC{$a\ \re\ a'$}
\RightLabel{$x\not\in\theta$}
\UnaryInfC{$b\matching{\mu}{(a,\hat{x})\Delta}\ \re\ b\matching{\mu}{(a',\hat{x})\Delta}$}
\DisplayProof\bigskip\\
\AxiomC{$a\ \re\ a'$}
\UnaryInfC{$b\matching{\mu}{(a,p_1\bullet p_2)\Delta}\ \re\ b\matching{\mu}{(a',p_1\bullet p_2)\Delta}$}
\DisplayProof
\end{center}

\caption{Context rules for matching-driven reduction}
\label{Contexts}
\end{figure}

\begin{figure}[p]
\begin{center}\vskip2mm
\begin{tabular}{r@{$\quad\rer\quad$}ll}
  $b\matching{\tau}{(a,\hat{x})\Delta}$ & $b^{\set{x\mapsto
      a}}\matching{\tau\cup\set{x}}{\Delta}$ & if $x\in\theta$, $x\not\in\tau$
  and $\fn{a}\cap\theta=\emptyset$\vspace{2mm}\\
$b\matching[\theta]{\theta}{\emptyset}$ & $b^\sigma$\\
$b\matching[\theta]{\tau}{\emptyset}$ & \fail & if $\tau\neq\theta$\\
$b\matching[\theta]{\fail}{\Delta}$ & \fail
\end{tabular}
\end{center}

\caption{Partial substitution rules}
\label{PartialSubstitutionRules}
\end{figure}

Moreover, it can be checked that the list structure of Example~\ref{ExampleList}
associated with the rules of Figure~\ref{EPPCRules} and the context rules of
Figure~\ref{Contexts} gives a deterministic reduction strategy for
\EPPC (which means that any term has at most one authorised redex).

\subsection{An Extension: Partial Substitution}
\label{EPMPS}

Relaxing the matching procedure generates new possibilities of evaluation,
which may bring more partial evaluation, more sharing or more parallelism. We
explore here an extension of \EPPC where the partial result of a matching can be
applied to the function body before the matching process is completed.

\begin{example}\label{PartialSubstitutionExample}
Consider the following reduction:\medskip\\
\begin{tabular}{ll}
\multicolumn{2}{l}{$(\case[x]{\hat{x}z}{((\case[\emptyset]{x}{b})\hat{c})})\,(\hat{c}t)$}\smallskip\\
$\reB$ & $(\case[x]{\hat{x}z}{(b\matching[\emptyset]{\emptyset}{(\hat{c},x)})})\,(\hat{c}t)$
\end{tabular}\medskip\\
The matching $\matching[\emptyset]{\emptyset}{(\hat{c},x)}$ is blocked because
of the presence of the variable $x$ in the pattern. Still, the external
application can be evaluated:\medskip\\
\begin{tabular}{ll}
$\reB$ & $(b\matching[\emptyset]{\emptyset}{(\hat{c},x)})\,\matching[x]{\emptyset}{(\hat{c}t,\hat{x}z)}$\\
$\reb^2$ & $(b\matching[\emptyset]{\emptyset}{(\hat{c},x)})\,\matching[x]{\emptyset}{(\hat{c}\bullet
    t,\hat{x}\bullet z)}$\\
$\rem$ & $(b\matching[\emptyset]{\emptyset}{(\hat{c},x)})\,\matching[x]{\emptyset}{(\hat{c},\hat{x})(t,z)}$\\
$\rem$ & $(b\matching[\emptyset]{\emptyset}{(\hat{c},x)})\,\matching[x]{\set{x\mapsto\hat{c}}}{(t,z)}$
\end{tabular}\medskip\\
Now, the external matching $\matching[x]{\set{x\mapsto\hat{c}}}{(t,z)}$ is also
blocked because of the variable $z$. However, its partial result is a
substitution for $x$ which, if applied, may unlock the internal
matching. Indeed, allowing this partial substitution could lead to a reduction
like:\medskip\\
\begin{tabular}{ll}
$\re$ & $(b\matching[\emptyset]{\emptyset}{(\hat{c},\hat{c})})\,\matching[x]{\set{x\mapsto\hat{c}}}{(t,z)}$\\
$\rem$ & $(b\matching[\emptyset]{\emptyset}{\emptyset})\,\matching[x]{\set{x\mapsto\hat{c}}}{(t,z)}$\\
$\rer$ & $b\matching[x]{\set{x\mapsto\hat{c}}}{(t,z)}$
\end{tabular}\medskip\\
where the internal matching is finally solved!
\end{example}

\noindent
This kind of power may be of interest in two situations:
\begin{itemize}
\item By allowing more reduction in open terms, we
  gain more partial evaluation capabilities. This may be interesting for
  greater sharing and efficient evaluation~\cite{Fuller}.
\item Suppose now that $z$ is replaced in the example by a possibly big term. In
  a parallel implementation we could complete the external matching and 
  evaluate the internal one in parallel. As pointed out in~\cite{MatchingINets},
  this might represent another gain in efficiency.
\end{itemize}

A light variation on \EPPC gives this new power to our formalism. The
principle of this variant is to systematically apply partial results
(substitutions) as soon as they are obtained. Hence they do not need to be
remembered in the object representing ongoing matching operations. Only a list
of used variables is remembered for linearity verification.

The object representing a matching is now $\matching{\tau}{\Delta}$ where $\tau$
is either \fail or the list of the names of the matching variables that have
already been used. Now the test of disjoint union of substitutions is replaced
by a simple test against $\tau$, while the final check compares $\theta$ and
$\tau$.

Initialisation, structural application, and most matching rules are the same in
this variant. The only differences are for the first matching rule and the
resolution rules, which are now as in Figure~\ref{PartialSubstitutionRules}.

Any \EPPC term can be translated into a term of this new calculus by applying
the following transformation: turn any $b\matching{\sigma}{\Delta}$ into
$b^\sigma\matching{dom(\sigma)}{\Delta}$ (there is nothing to change in a failed
matching).

The simulation between \EPPC and this extension is only one way: any reduction
of \EPPC is mapped by the previous morphism to a reduction sequence, but the
converse is not true. Indeed the calculus with partial substitution allows new
reductions, as pictured in Example~\ref{PartialSubstitutionExample}. Confluence
for this variant seems to be provable using the same technique as for plain
\EPPC.

\section*{Conclusion}

The {\em Pure Pattern Calculus} is a compact framework modelling pattern
matching with dynamic patterns. However, the conciseness of \PPC is due to the
use of several meta-level notions which deepens the gap between the calculus and
implementation-related problems. This contribution defines the {\em Pure Pattern
  Calculus with Explicit Matching}, a refinement which is confluent and
simulates \PPC, and allows reasoning on the pattern matching mechanisms.

This enables a very simple definition of new reduction strategies in the spirit
of call-by-name, which is new in this kind of framework since the reduction of
the argument of a function depends on the pattern of the function, pattern which
is itself a dynamic object.
In the same direction, it would be interesting to express standardisation in
pattern calculi (as presented for example in~\cite{Standard}) using explicit matching.

\bibliographystyle{eptcs}

\end{document}